%% file: full.tex
\def\year{2017}\relax
\documentclass[letterpaper]{article}
\usepackage{aaai17}
\usepackage{hyperref}
\usepackage{times}
\usepackage{helvet}
\usepackage{courier}
\usepackage{amsmath,amssymb,amsthm,bm,subfig}
\usepackage{graphicx}

\usepackage{url}

\usepackage{algorithmic}
\usepackage{algorithm}
\renewcommand{\algorithmicrequire}{\textbf{Input:}}
\renewcommand{\algorithmicensure}{\textbf{Assumption:}}
\renewcommand\algorithmiccomment[1]{%
    \qquad \textcolor{gray}{\# #1}%
}

\usepackage[dvipsnames]{xcolor}
\usepackage{color,colortbl}
\definecolor{Gray}{gray}{0.9}

\newtheorem{theorem}{Theorem}
\newtheorem{lemma}[theorem]{Lemma}

\newtheorem{claim}[theorem]{Claim}

\newtheorem{remark}[theorem]{Remark}
\theoremstyle{definition}

\newcommand{\bbN}{\mathbb{N}}
\newcommand{\bbR}{\mathbb{R}}
\newcommand{\bbZ}{\mathbb{Z}}
\newcommand{\bB}{{\bm B}}

\newcommand{\bo}{{\bm o}}

\newcommand{\bx}{{\bm x}}
\newcommand{\by}{{\bm y}}

\newcommand{\bchi}{{\bm \chi}}
\newcommand{\bzero}{{\bm 0}}

\newcommand{\caB}{\mathcal{B}}
\newcommand{\caD}{\mathcal{D}}
\newcommand{\caG}{\mathcal{G}}
\newcommand{\caP}{\mathcal{P}}
\newcommand{\set}[1]{\{#1\}}
\newcommand{\E}{\mathop{\mathbf{E}}}
\newcommand{\ceil}[1]{\lceil #1 \rceil}
\newcommand{\floor}[1]{\lfloor #1 \rfloor}
\newcommand{\abs}[1]{\lvert#1\rvert}
\newcommand{\norm}[1]{\lVert#1\rVert}

\begin{document}
\nocopyright

\title{Non-monotone DR-Submodular Function Maximization\\ (Full version)}

\author{
  Tasuku Soma \\
  The University of Tokyo \\
  tasuku\_soma@mist.i.u-tokyo.ac.jp
  \And
  Yuichi Yoshida\\
  National Institute of Informatics, \emph{and}\\
  Preferred Infrastructure, Inc.\\
  yyoshida@nii.ac.jp
}
\maketitle

\begin{abstract}
  We consider non-monotone DR-submodular function maximization, where DR-submodularity (diminishing return submodularity) is an extension of submodularity for functions over the integer lattice based on the concept of the diminishing return property.
  Maximizing non-monotone DR-submodular functions has many applications in machine learning that cannot be captured by submodular set functions.
  In this paper, we present a $\frac{1}{2+\epsilon}$-approximation algorithm with a running time of roughly $O(\frac{n}{\epsilon}\log^2 B)$, where $n$ is the size of the ground set, $B$ is the maximum value of a coordinate, and $\epsilon > 0$ is a parameter.
  The approximation ratio is almost tight and the dependency of running time on $B$ is exponentially smaller than the naive greedy algorithm.
  Experiments on synthetic and real-world datasets demonstrate that our algorithm outputs almost the best solution compared to other baseline algorithms, whereas its running time is several orders of magnitude faster.
\end{abstract}



\section{Introduction}

Submodular functions have played a key role in various tasks in machine learning, statistics, social science, and economics.
A set function $f : 2^E \to \bbR$ with a ground set $E$ is \emph{submodular} if
\begin{align*}
  f(X \cup \set{e}) - f(X) \geq f(Y \cup \set{e}) - f(Y)
\end{align*}
for arbitrary sets $X,Y \subseteq E$ with $X \subseteq Y$, and an element $e \in E \setminus Y$.
The importance and usefulness of submodularity in these areas are due to the fact that submodular functions naturally capture the \emph{diminishing return property}.
Various important functions in these areas such as the entropy function, coverage function, and utility functions satisfy this property. See, e.g.,~\cite{Krause2014survey,Fujishige2005}.

Recently, \emph{maximizing} (non-monotone) submodular functions has attracted particular interest in the machine learning community.
In contrast to \emph{minimizing} submodular functions, which can be done in polynomial time, maximizing submodular functions is NP-hard in general.
However, we can achieve a constant factor approximation for various settings.
Notably,~\cite{Buchbinder:2012hc} presented a very elegant \emph{double greedy} algorithm for (unconstrained) submodular function maximization, which was the first algorithm achieving $\frac{1}{2}$-approximation, and this approximation ratio is tight~\cite{Feige2011}.
Applications of non-monotone submodular function maximization include efficient sensor placement~\cite{Krause2008a}, privacy in online services~\cite{Krause2008utility}, and maximum entropy sampling~\cite{Ko:1995hk}.

The models and applications mentioned so far are built upon submodular \emph{set functions}.
Although set functions are fairly powerful for describing problems such as variable selection, we sometimes face difficult situations that cannot be cast with set functions.
For example, in the budget allocation problem~\cite{Alon2012}, we would like to decide how much budget should be set aside for each ad source, rather than whether we use the ad source or not.
A similar issue arises when we consider models allowing multiple choices of an element in the ground set.

To deal with such situations, several generalizations of submodularity have been proposed.
\cite{Soma:2014tp} devised a general framework for maximizing \emph{monotone submodular functions on the integer lattice}, and showed that the budget allocation problem and its variant fall into this framework.
In their framework, functions are defined over the integer lattice $\bbZ_+^E$ and therefore effectively represent discrete allocations of budget.
Regarding the original motivation for the diminishing return property, one can naturally consider its generalization to the integer lattice:
a function $f : \bbZ_+^E \to \bbR$ satisfying
\begin{align*}
  f(\bx + \bchi_e) - f(\bx) \geq f(\by + \bchi_e) - f(\by)
\end{align*}
for $\bx \leq \by$ and $e \in E$, where $\bchi_e \in \bbR^E$ is the vector with $\bchi_e(e) = 1$ and $\bchi_e(a) = 0$ for every $a \neq e$.
Such functions are called \emph{diminishing return submodular (DR-submodular) functions}~\cite{sfcover:nips2015} or \emph{coordinate-wise concave submodular functions}~\cite{Milgrom2009}.
DR-submodular functions have found various applications in generalized sensor placement~\cite{sfcover:nips2015} and (a natural special case of) the budget allocation problem~\cite{Soma:2014tp}.

As a related notion,
a function is said to be \emph{lattice submodular} if
\[
f(\bx) + f(\by) \geq f(\bx \vee \by) + f(\bx \wedge \by)
\]
for arbitrary $\bx$ and $\by$, where $\vee$ and $\wedge$ are coordinate-wise max and min, respectively.
Note that DR-submodularity is \emph{stronger} than \emph{lattice submodularity} in general (see, e.g., \cite{Soma:2014tp}).
Nevertheless, we consider the DR-submodularity to be a ``natural definition'' of submodularity, at least for the applications mentioned so far, because the diminishing return property is crucial in these real-world scenarios.

\subsection{Our contributions}
We design a novel polynomial-time approximation algorithm for maximizing (non-monotone) DR-submodular functions.
More precisely, we consider the optimization problem
\begin{align}\label{eq:max-DRsfm}
  \begin{array}{ll}
    \text{maximize}& f(\bx)\\
    \text{subject to} & \bzero \leq \bx \leq \bB,
  \end{array}
\end{align}
where $f : \bbZ_+^E \to \bbR_+$ is a non-negative DR-submodular function, $\bzero$ is the zero vector, and $\bB \in \bbZ_+^E$ is a vector representing the maximum values for each coordinate.
When $\bB$ is the all-ones vector, this is equivalent to the original (unconstrained) submodular function maximization.
We assume that $f$ is given as an evaluation oracle; when we specify $\bx \in \bbZ_+^E$, the oracle returns the value of $f(\bx)$.

Our algorithm achieves $\frac{1}{2+\epsilon}$-approximation for any constant $\epsilon > 0$ in $O(\frac{|E|}{\epsilon} \cdot \log (\frac{\Delta}{\delta})\log B \cdot (\theta + \log B))$ time,
where $\delta$ and $\Delta$ are the minimum positive marginal gain and maximum positive values, respectively, of $f$, $B = \norm{\bB}_{\infty} := \max_{e \in E}\bB(e)$, and $\theta$ is the running time of evaluating (the oracle for) $f$.
To our knowledge, this is the first polynomial-time algorithm achieving (roughly) $\frac{1}{2}$-approximation.

We also conduct numerical experiments on the revenue maximization problem using real-world networks.
The experimental results show that the solution quality of our algorithm is comparable to other algorithms.
Furthermore, our algorithm runs several orders of magnitude faster than other algorithms when $B$ is large.

DR-submodularity is necessary for obtaining polynomial-time algorithms with a meaningful approximation guarantee;
if $f$ is only lattice submodular, then we cannot obtain constant-approximation in polynomial time.
To see this, it suffices to observe that an \emph{arbitrary} univariate function is lattice submodular, and therefore finding an (approximate) maximum value must invoke $O(B)$ queries.
We note that representing an integer $B$ requires $\ceil{\log_2 B}$ bits.
Hence, the running time of $O(B)$ is pseudopolynomial rather than polynomial.


\subsection*{Fast simulation of the double greedy algorithm}
Naturally, one can reduce the problem~\eqref{eq:max-DRsfm} to maximization of a submodular set function
by simply duplicating each element $e$ in the ground set into $\bB(e)$ distinct copies and defining a set function over the set of all the copies.
One can then run the double greedy algorithm~\cite{Buchbinder:2012hc} to obtain $\frac{1}{2}$-approximation.
This reduction is simple but has one large drawback; the size of the new ground set is $\sum_{e\in E} \bB(e)$, a pseudopolynomial in $B$.
Therefore, this naive double greedy algorithm does not scale to a situation where $B$ is large.

For scalability, we need an additional trick that reduces the pseudo-polynomial running time to a polynomial one.
In \emph{monotone} submodular function maximization on the integer lattice,~\cite{sfcover:nips2015,sfm:arxiv2015} provide such a speedup trick, which effectively combines the \emph{decreasing threshold technique}~\cite{Badanidiyuru2013} with binary search.
However, a similar technique does not apply to our setting, because the double greedy algorithm works differently from (single) greedy algorithms for monotone submodular function maximization.
The double greedy algorithm examines each element in a \emph{fixed} order and marginal gains are used to decide whether to include the element or not.
In contrast, the greedy algorithm chooses each element in \emph{decreasing} order of marginal gains, and this property is crucial for the decreasing threshold technique.

We resolve this issue by splitting the set of all marginal gains into polynomially many small intervals.
For each interval, we approximately execute multiple steps of the double greedy algorithm at once, as long as the marginal gains remain in the interval.
Because the marginal gains do not change (much) within the interval, this simulation can be done with polynomially many queries and polynomial-time overhead.
To our knowledge, this speedup technique is not known in the literature and is therefore of more general interest.

Very recently,~\cite{DBLP:journals/corr/EneN16} pointed out that a DR-submodular function $f:\set{0,1,\ldots,B}^E \to \bbR_+$ can be expressed as a submodular set function $g$ over a polynomial-sized ground set, which turns out to be $E \times \set{0,1,\ldots,k-1}$, where $k=\ceil{\log_2(B+1)}$.
Their idea is representing $\bx(e)$ in binary form for each $e \in E$, and bits in the binary representations form the new ground set.
We may want to apply the double greedy algorithm to $g$ in order to get a polynomial-time approximation algorithm.
However, this strategy has the following two drawbacks: (i) The value of $g(E \times \set{0,1,\ldots,k-1})$ is defined as $f(\bx)$, where $\bx(e) = 2^k-1$ for every $e \in E$. This means that we have to extend the domain of $f$.
(ii) More crucially, the double greedy algorithm on $g$ may return a large set such as $E \times \set{0,1,\ldots,k-1}$ whose corresponding vector $\bx \in \bbZ_+^E$ may violate the constraint $\bx \leq \bB$.
Although we can resolve these issues by introducing a knapsack constraint, it is not a practical solution because existing algorithms for knapsack constraints~\cite{Lee:2009tc,Chekuri:2014ed} are slow and have worse approximation ratios than $1/2$.


\paragraph{Notations}
For an integer $n \in \bbN$, $[n]$ denotes the set $\set{1,\ldots,n}$.
For vectors $\bx, \by \in \bbZ^E$, we define $f(\bx \mid \by) := f(\bx + \by) - f(\by)$.
The $\ell_1$-norm and $\ell_{\infty}$-norm of a vector $\bx \in \bbZ^E$ are defined as $\norm{\bx}_1 := \sum_{e \in E} \abs{\bx(e)}$ and $\norm{\bx}_{\infty} := \max_{e \in E} \abs{\bx(e)}$, respectively.



\section{Related work}
As mentioned above, there have been many efforts to maximize submodular functions on the integer lattice.
Perhaps the work most related to our interest is~\cite{Gottschalk:2015fq},
in which the authors considered maximizing lattice submodular functions over the bounded integer lattice and designed $1\over3$-approximation pseudopolynomial-time algorithm.
Their algorithm was also based on the double greedy algorithm, but does not include a speeding up technique, as proposed in this paper.

In addition there are several studies on the \emph{constrained} maximization of submodular functions~\cite{Feige2011,Buchbinder2014,Buchbinder2015}, although we focus on the unconstrained case.
Many algorithms for maximizing submodular functions are randomized, but a very recent work~\cite{Buchbinder2015} devised a derandomized version of the double greedy algorithm.
\cite{Gotovos2015} considered maximizing non-monotone submodular functions in the \emph{adaptive} setting, a concept introduced in~\cite{Golovin:2011cn}.

A continuous analogue of DR-submodular functions is considered in~\cite{DBLP:journals/corr/BianMB016}.


\section{Algorithms}\label{sec:algorithm}

In this section, we present a polynomial-time approximation algorithm for maximizing (non-monotone) DR-submodular functions.
We first explain a simple adaption of the double greedy algorithm for (set) submodular functions to our setting, which runs in pseudopolynomial time.
Then, we show how to achieve a polynomial number of oracle calls.
Finally, we provide an algorithm with a polynomial running time (details are placed in Appendix~\ref{sec:proof-truely-polynomial}).

\subsection{Pseudopolynomial-time algorithm}\label{subsec:pseudo-poly}

Algorithm~\ref{alg:pseudo-poly} is an immediate extension of the double greedy algorithm for maximizing submodular (set) functions~\cite{Buchbinder:2012hc} to our setting.
We start with $\bx = \bzero$ and $\by = \bB$, and then for each $e \in E$, we tighten the gap between $\bx(e)$ and $\by(e)$ until they become exactly the same.
Let $\alpha = f(\bchi_e \mid \bx)$ and $\beta = f(-\bchi_e \mid \by)$.
We note that
\[
\alpha + \beta = f(\bx + \bchi_e) - f(\bx) - (f(\by) - f(\by - \bchi_e)) \geq 0
\]
holds from the DR-submodularity of $f$.
Hence, if $\beta < 0$, then $\alpha > 0$ must hold, and we increase $\bx(e)$ by one.
Similarly, if $\alpha < 0$, then $\beta > 0$ must hold, and we decrease $\by(e)$ by one.
When both of them are non-negative, we increase $\bx(e)$ by one with probability $\frac{\alpha}{\alpha+\beta}$, or decrease $\by(e)$ by one with the complement probability $\frac{\beta}{\alpha+\beta}$.

\begin{algorithm}[t]
  \caption{Pseudopolynomial-time algorithm}\label{alg:pseudo-poly}
  \begin{algorithmic}[1]
    \REQUIRE{$f:\bbZ_+^E \to \bbR_+$, $\bB \in \bbZ_+^E$.}
    \ENSURE{$f$ is DR-submodular.}
    \STATE $\bx \leftarrow \bzero$, $\by \leftarrow \bB$.
    \FOR{$e \in E$}\label{line:pseudo-poly-e}
      \WHILE{$\bx(e) < \by(e)$}
        \STATE $\alpha \leftarrow f(\bchi_{e} \mid \bx)$ and $\beta \leftarrow f(-\bchi_{e} \mid \by )$. \label{line:pseudo-poly-a-b}
        \IF{$\beta < 0$}
          \STATE{$\bx(e) \leftarrow \bx(e) + 1$.}
        \ELSIF{$\alpha < 0$}
          \STATE  $\by(e) \leftarrow \by(e) - 1$.
        \ELSE
          \STATE{$\bx(e) \leftarrow \bx(e) + 1$ with probability $\frac{\alpha}{\alpha+\beta}$ and $\by(e) \leftarrow \by(e) - 1$ with the complement probability $\frac{\beta}{\alpha+\beta}$. If $\alpha = \beta = 0$, we assume $\frac{\alpha}{\alpha+\beta}=1$.}\label{line:pseudo-poly-else}
        \ENDIF
      \ENDWHILE
    \ENDFOR
    \RETURN $\bx$.
  \end{algorithmic}
\end{algorithm}

\begin{theorem}
  Algorithm~\ref{alg:pseudo-poly} is a $\frac{1}{2}$-approximation algorithm for~\eqref{eq:max-DRsfm} with time complexity $O(\|\bB\|_1 \cdot \theta + \|\bB\|_1)$,
  where $\theta$ is the running time of evaluating $f$.
\end{theorem}
We omit the proof as it is a simple modification of the analysis of the original algorithm.

\subsection{Algorithm with polynomially many oracle calls}\label{subsec:poly-oracle-call}

In this section, we present an algorithm with a  polynomial number of oracle calls.

Our strategy is to simulate Algorithm~\ref{alg:pseudo-poly} without evaluating the input function $f$ many times.
A key observation is that, at Line~\ref{line:pseudo-poly-a-b} of Algorithm~\ref{alg:pseudo-poly}, we do not need to know the exact value of $f(\bchi_e \mid \bx)$ and $f(-\bchi_e \mid \by)$; good approximations to them are sufficient to achieve an approximation guarantee close to $\frac{1}{2}$.
To exploit this observation, we first design an algorithm that outputs (sketches of) approximations to the functions $g(b) := f(\bchi_e \mid \bx + b\bchi_e)$ and $h(b) := f(-\bchi_e \mid \by-b\bchi_e)$.
Note that $g$ and $h$ are non-increasing functions in $b$ because of the DR-submodularity of $f$.

To illustrate this idea, let us consider a non-increasing function $\phi:\set{0,1,\ldots,B-1}\to \bbR$ and suppose that $\phi$ is non-negative ($\phi$ is either $g$ or $h$ later on).
Let $\delta$ and $\Delta$ be the minimum and the maximum positive values of $\phi$, respectively.
Then, for each $\delta \leq \tau \leq \Delta$ of the form $\delta (1+\epsilon)^k$, we find the minimum $b_\tau$ such that $\phi(b_\tau) < \tau$ (we regard $\phi(B) = -\infty$).
From the non-increasing property of $\phi$, we then have $\phi(b) \geq \tau$ for any $b < b_\tau$.
Using the set of pairs $\set{(\tau,b_\tau)}_{\tau}$, we can obtain a good approximation to $\phi$.
The details are provided in Algorithm~\ref{alg:approximate}.

\begin{algorithm}[t]
  \caption{Sketching subroutine for Algorithm~\ref{alg:polynomial-oracle}}\label{alg:approximate}
  \begin{algorithmic}[1]
    \REQUIRE{$\phi:\set{0,1,\ldots,B-1} \to \bbR$, $\epsilon > 0$.}
    \ENSURE{$\phi$ is non-increasing.}
    \STATE $S \leftarrow \emptyset$.
    We regard $\phi(B) = -\infty$.
    \STATE Find the minimum $b_0 \in \set{0,1,\ldots,B}$ with $\phi(b_0) \leq 0$ by binary search. \label{line:approximate-b}
    \IF[{$\phi$ has a positive value.}]{$b_0 \geq 1$}
      \STATE $\Delta \leftarrow \phi(0)$ and $\delta \leftarrow \phi(b_0-1)$.
      \FOR{($\tau \leftarrow \delta$; $\tau \leq \Delta$; $\tau \leftarrow (1+\epsilon)\tau$)}
        \STATE Find the minimum $b_\tau \in \set{0,1,\ldots,B}$ with $\phi(b_\tau) < \tau$ by binary search.
        \STATE $S \leftarrow S \cup \set{(b_\tau,\tau)}$
      \ENDFOR
        \IF{$b_\delta \neq B$}
        \STATE{$S \leftarrow S \cup \set{(B,0)}$.}
        \ENDIF
      \ELSE[$\phi$ is non-positive.]
      \STATE $S \leftarrow S \cup \set{(B,0)}$.
    \ENDIF
    \RETURN $S$.
  \end{algorithmic}
\end{algorithm}

\begin{lemma}\label{lem:approximate}
  For any $\phi:\set{0,1,\ldots,B-1}\to \bbR$ and $\epsilon > 0$, Algorithm~\ref{alg:approximate} outputs a set of pairs $\set{(b_\tau,\tau)}_\tau$ from which, for any $b \in \set{0,1,\ldots,B-1}$, we can reconstruct a value $v$ in $O(\log B)$ time such that $v \leq \phi(b) < (1+\epsilon)v$ if $\phi(b) > 0$ and $v = 0$ otherwise.
  The time complexity of Algorithm~\ref{alg:approximate} is $O(\frac{1}{\epsilon}\log (\frac{\Delta}{\delta}) \log B \cdot \theta)$ if $\phi$ has a positive value, where $\delta$ and $\Delta$ are the minimum and maximum positive values of $\phi$, respectively, and $\theta$ is the running time of evaluating $\phi$, and is $O(\log B \cdot \theta)$ otherwise.
\end{lemma}
\begin{proof}
  Let $S = \set{(b_\tau,\tau)}_\tau$ be the set of pairs output by Algorithm~\ref{alg:approximate}.
  Our reconstruction algorithm is as follows:
  Given $b \in \set{0,1,\ldots,B-1}$, let $(b_{\tau^*},\tau^*)$ be the pair with the minimum $b_{\tau^*}$, where $b < b_{\tau^*}$.
  Note that such a $b_{\tau^*}$ always exists because a pair of the form $(B,\cdot)$ is always added to $S$.
  We then output $\tau^*$.
  The time complexity of this reconstruction algorithm is clearly $O(\log B)$.

  We now show the correctness of the reconstruction algorithm.
  If $\phi(b) > 0$, then, in particular, we have $\phi(b) \geq \delta$.
  Then, $\tau^*$ is the maximum value of the form $\delta(1+\epsilon)^k$ at most $\phi(b)$.
  Hence, we have $\tau^* \leq \phi(b) < (1+\epsilon)\tau^*$.
  If $\phi(b) \leq 0$, $(b_{\tau^*},\tau^*) = (B,0)$ and we output zero.

  Finally, we analyze the time complexity of Algorithm~\ref{alg:approximate}.
  Each binary search requires $O(\log B)$ time.
  The number of  binary searches performed is $O(\log_{1+\epsilon}\frac{\Delta}{\delta} ) = O(\frac{1}{\epsilon}\log\frac{\Delta}{\delta})$ when $\phi$ has a positive value and 1 when $\phi$ is non-positive.
  Hence, we have the desired time complexity.
\end{proof}
We can regard Algorithm~\ref{alg:approximate} as providing a value oracle for a function $\tilde{\phi}:\set{0,1,\ldots,B-1}\to \bbR_+$ that is an approximation to the input function $\phi:\set{0,1,\ldots,B-1}\to \bbR$.

We now describe our algorithm for maximizing DR-submodular functions.
The basic idea is similar to Algorithm~\ref{alg:pseudo-poly}, but when we need $f(\bchi_e \mid \bx)$ and $f(-\bchi_e \mid \by)$, we use approximations to them instead.
Let $\alpha$ and $\beta$ be approximations to $f(\bchi_e \mid \bx)$ and $f(-\bchi_e \mid \by)$, respectively, obtained by Algorithm~\ref{alg:approximate}.
Then, we increase $\bx(e)$ by one with probability $\frac{\alpha}{\alpha+\beta}$ and decrease $\by(e)$ by one with the complement probability $\frac{\beta}{\alpha+\beta}$.
The details are given in Algorithm~\ref{alg:polynomial-oracle}.

\begin{algorithm}[t]
  \caption{Algorithm with polynomially many queries}\label{alg:polynomial-oracle}
  \begin{algorithmic}[1]
    \REQUIRE{$f:\bbZ_+^E \to \bbR_+$, $\bB \in \bbZ_+^E$, $\epsilon > 0$.}
    \ENSURE{$f$ is DR-submodular.}
    \STATE $\bx \leftarrow \bzero$, $\by \leftarrow \bB$.
    \FOR{$e \in E$}
      \STATE Define $g,h:\set{0,1,\ldots,B-1} \to \bbR$ as $g(b) = f(\bchi_e \mid \bx+b\bchi_e)$ and $h(b) = f(- \bchi_e \mid \by-b \bchi_e)$.\hspace{-1em} %
      \STATE Let $\tilde{g}$ and $\tilde{h}$ be approximations to $g$ and $h$, respectively, given by Algorithm~\ref{alg:approximate}.
      \WHILE{$\bx(e) < \by(e)$} \label{line:polynomial-oracle-while}
        \STATE $\alpha \leftarrow \tilde{g}(\bx(e))$ and $\beta \leftarrow \tilde{h}(\bB(e) - \by(e))$. \label{line:alpha-beta}
        \STATE{$\bx(e) \leftarrow \bx(e) + 1$ with probability $\frac{\alpha}{\alpha+\beta}$ and $\by(e) \leftarrow \by(e) - 1$ with the complement probability $\frac{\beta}{\alpha+\beta}$.
        If $\alpha = \beta = 0$, we assume $\frac{\alpha}{\alpha+\beta}=1$.}
      \ENDWHILE
    \ENDFOR
    \RETURN $\bx$.
  \end{algorithmic}
\end{algorithm}

We now analyze Algorithm~\ref{alg:polynomial-oracle}.
An \emph{iteration} refers to an iteration in the while loop from Line~\ref{line:polynomial-oracle-while}.
We have $\|\bB\|_1$ iterations in total.
For $k \in \set{1,\ldots,\|\bB\|_1}$, let $\bx_k$ and $\by_k$ be $\bx$ and $\by$, respectively, right after the $k$th iteration.
Note that $\bx_{\|\bB\|_1} = \by_{\|\bB\|_1}$ is the output of the algorithm.
We define $\bx_0 = \bzero$ and $\by_0 = \bB$ for convenience.

Let $\bo$ be an optimal solution.
For $k \in \set{0,1,\ldots,\|\bB\|_1}$, we then define $\bo_{k} = (\bo \vee \bx_{k}) \wedge \by_{k}$.
Note that $\bo_0 = \bo$ holds and $\bo_{\|\bB\|_1}$ equals the output of the algorithm.
We have the following key lemma.
\begin{lemma}\label{lem:key-lemma}
  For every $k \in [\|\bB\|_1]$, we have
  \begin{align}
      &\E[f(\bo_{k-1}) - f(\bo_{k})] \nonumber\\
    &\leq  \frac{1+\epsilon}{2}\E[f(\bx_{k}) - f(\bx_{k-1}) + f(\by_{k}) - f(\by_{k-1}) ] \label{eq:key-inequality}
  \end{align}
\end{lemma}
\begin{proof}
  Fix $k \in [\|\bB\|_1]$ and let $e$ be the element of interest in the $k$th iteration.
  Let $\alpha$ and $\beta$ be the values in Line~\ref{line:alpha-beta} in the $k$th iteration.
  We then have
  \begin{align}
    & \E[f(\bx_{k}) - f(\bx_{k-1}) + f(\by_{k}) - f(\by_{k-1})] \nonumber \\
    &=
    \frac{\alpha}{\alpha+\beta}f(\bchi_{e} \mid \bx_{k-1})
    +
    \frac{\beta}{\alpha+\beta}f(-\bchi_{e} \mid \by_{k-1}) \nonumber \\
    &\geq
    \frac{\alpha}{\alpha+\beta}\alpha
    +
    \frac{\beta}{\alpha+\beta}\beta
    = \frac{\alpha^2+\beta^2}{\alpha+\beta}, \label{eq:marginal-of-alg}
    \end{align}
    where we use the guarantee in Lemma~\ref{lem:approximate} in the inequality.

    We next establish an upper bound of $\E[f(\bo_{k-1}) - f(\bo_{k})]$.
    As $\bo_k = (\bo \vee \bx_k) \wedge \by_k$, conditioned on a fixed $\bo_{k-1}$, we obtain
    \begin{align}
        &\E[f(\bo_{k-1}) - f(\bo_{k})] \nonumber \\
      &=
      \frac{\alpha}{\alpha+\beta} \Bigl(f(\bo_{k-1})-f(\bo_{k-1} \vee \bx_{k}(e) \bchi_e)\Bigr) \nonumber \\
      & \qquad +
      \frac{\beta}{\alpha+\beta} \Bigl(f(\bo_{k-1})-f(\bo_{k-1} \wedge \by_{k}(e) \bchi_e)\Bigr). \label{eq:marginal-of-opt}
    \end{align}
    \begin{claim}\label{cla:key-lemma}
      $\eqref{eq:marginal-of-opt} \leq \frac{(1+\epsilon)\alpha\beta}{\alpha+\beta}$.
    \end{claim}
    \begin{proof}
      We show this claim by considering the following three cases.

      If $\bx_{k}(e) \leq \bo_{k-1}(e) \leq \by_{k}(e)$,
      then~\eqref{eq:marginal-of-opt} is zero.

      If $\bo_{k-1}(e) < \bx_{k}(e)$, then $\bo_{k}(e) = \bo_{k-1}(e) + 1$, and the first term of~\eqref{eq:marginal-of-opt} is
      \begin{align*}
        & f(\bo_{k-1})-f(\bo_{k-1} \vee \bx_{k}(e) \bchi_e) \\
        & = f(\bo_{k-1}) - f(\bo_{k-1}+\bchi_e)\\
        & \leq f(\by_{k-1} - \bchi_e) - f(\by_{k-1}) \\
        & = f(-\bchi_e \mid \by_{k-1}) \\
        & \leq (1+\epsilon)\beta.
      \end{align*}
      Here, the first inequality uses the DR-submodularity of $f$ and the fact that $\bo_{k-1} \leq \by_{k-1}-\bchi_e$,
      and the second inequality uses the guarantee in Lemma~\ref{lem:approximate}.
      The second term of~\eqref{eq:marginal-of-opt} is zero, and hence we have $\eqref{eq:marginal-of-opt} \leq \frac{(1+\epsilon)\alpha\beta}{\alpha+\beta}$.

      If $\by_{k}(e) < \bo_{k-1}(e)$, then by a similar argument, we have $\eqref{eq:marginal-of-opt} \leq \frac{(1+\epsilon)\alpha\beta}{\alpha+\beta}$.
    \end{proof}
    We now return to proving Lemma~\ref{lem:key-lemma}.
    By Claim~\ref{cla:key-lemma}, 
    \[
      \eqref{eq:marginal-of-opt}
      \leq
      \frac{(1+\epsilon)\alpha\beta}{\alpha+\beta} \leq  \frac{1+\epsilon}{2}\Bigl( \frac{\alpha^2+\beta^2}{\alpha+\beta} \Bigr)
      \leq
       \frac{1+\epsilon}{2}\cdot \eqref{eq:marginal-of-alg},
    \]
    which indicates the desired result.
\end{proof}

\begin{theorem}\label{the:polynomial-oracle}
  Algorithm~\ref{alg:polynomial-oracle} is a $\frac{1}{2+\epsilon}$-approximation algorithm for~\eqref{eq:max-DRsfm} with time complexity $O(\frac{|E|}{\epsilon} \cdot \log (\frac{\Delta}{\delta}) \log \|\bB\|_\infty \cdot \theta + \|\bB\|_1\log \|\bB\|_\infty)$,
  where $\delta$ and $\Delta$ are the minimum positive marginal gain and the maximum positive value, respectively, of $f$ and $\theta$ is the running time of evaluating $f$.
\end{theorem}
\begin{proof}
  Summing up~\eqref{eq:key-inequality} for $k \in [\|\bB\|_1]$, we get
  \begin{align*}
    & \sum_{k = 1}^{\|\bB\|_1} \E[f(\bo_{k-1}) - f(\bo_{k})] \\
    &\leq
     \frac{1+\epsilon}{2}\sum_{k = 1}^{\|\bB\|_1}\E[f(\bx_{k}) - f(\bx_{k-1}) + f(\by_{k}) - f(\by_{k-1}) ].
  \end{align*}
  The above sum is telescopic, and hence we obtain
  \begin{align*}
    & \E[f(\bo_0) - f(\bo_{\|\bB\|_1})] \\
    &\leq
    \frac{1+\epsilon}{2}\E[ f(\bx_{\|\bB\|_1}) - f(\bx_0) + f(\by_{\|\bB\|_1}) - f(\by_0)] \\
    &\leq
    \frac{1+\epsilon}{2} \E[f(\bx_{\|\bB\|_1})+f(\by_{\|\bB\|_1})] \\
    & = (1+\epsilon) \E[f(\bx_{\|\bB\|_1})].
  \end{align*}
  The second inequality uses the fact that $f$ is non-negative, and the last equality uses $\by_{\|\bB\|_1} = \bx_{\|\bB\|_1}$.
  Because $\E[f(\bo_0) - f(\bo_{\|\bB\|_1})] = f(\bo) - \E[f(\bx_{\|\bB\|_1})]$,
  we obtain that $\E[f(\bx_{\|\bB\|_1})] \geq \frac{1}{2+\epsilon}f(\bo)$.

  We now analyze the time complexity.
  We only query the input function $f$ inside of Algorithm~\ref{alg:approximate}, and the number of oracle calls is $O(\frac{|E|}{\epsilon} \log (\frac{\Delta}{\delta}) \log B)$ by Lemma~\ref{lem:approximate}.
  Note that we invoke Algorithm~\ref{alg:approximate} with $g$ and $h$, and the minimum positive values of $g$ and $h$ are at least the minimum positive marginal gain $\delta$ of $f$.
  The number of iterations is $\|\bB\|_1$, and we need $O(\log B)$ time to access $\tilde{g}$ and $\tilde{h}$.
  Hence, the total time complexity is as stated.
\end{proof}

\begin{remark}\label{rem:non-negative}
  We note that even if $f$ is not a non-negative function, the proof of Theorem~\ref{the:polynomial-oracle} works as long as $f(\bx_0) \geq 0$ and $f(\by_0) \geq 0$, that is, $f(\bzero) \geq 0$ and $f(\bB) \geq 0$.
  Hence, given a DR-submodular function $f:\bbZ_+^E \to \bbR$ and $\bB \in \bbZ_+^E$, we can obtain a $\frac{1}{2+\epsilon}$-approximation algorithm for the following problem:
  \begin{align}
    \begin{array}{ll}
      \text{maximize} & f(\bx) - \min\set{f(\bzero),f(\bB)} \\
      \text{subject to} & \bzero \leq \bx \leq \bB,
    \end{array}\label{eq:max-DRsfm-g}
  \end{align}
  This observation is useful, as the objective function often takes negative values in real-world applications.
\end{remark}

\subsection{Polynomial-time algorithm}\label{subsec:poly}

In many applications, the running time needed to evaluate the input function is a bottleneck, and hence Algorithm~\ref{alg:polynomial-oracle} is already satisfactory.
However, it is theoretically interesting to reduce the total running time to a polynomial, and we show the following.
The proof is deferred to Appendix~\ref{sec:proof-truely-polynomial}.
\begin{theorem}\label{the:truely-polynomial}
  There exists
  a $\frac{1}{2+2\epsilon}$-approximation algorithm with time complexity $\widetilde{O}(\frac{|E|}{\epsilon}  \log (\frac{\Delta}{\delta})\log \|\bB\|_\infty \cdot (\theta + \log \|\bB\|_\infty))$,
  where $\delta$ and $\Delta$ are the minimum positive marginal gain and the maximum positive value, respectively of $f$ and $\theta$ is the running time of evaluating $f$.
  Here $\widetilde{O}(T)$ means $O(T \log^c T)$ for some $c \in \bbN$.
\end{theorem}


\section{Experiments}\label{sec:experiments}
In this section, we show our experimental results and the superiority of our algorithm with respect to other baseline algorithms.

\subsection{Experimental setting}

We conducted experiments on a Linux server with an Intel Xeon E5-2690 (2.90 GHz) processor and 256 GB of main memory.
All the algorithms were implemented in C\# and were run using Mono 4.2.3.

We compared the following four algorithms:
\begin{itemize}
\itemsep=0pt
\item Single Greedy (\textsf{SG}): We start with $\bx = \bzero$. For each element $e \in E$, as long as the marginal gain of adding $\bchi_e$ to the current solution $\bx$ is positive, we add it to $\bx$.
The reason that we do not choose the element with the maximum marginal gain is to reduce the number of oracle calls, and our preliminary experiments showed that such a tweak does not improve the solution quality.
\item Double Greedy (\textsf{DG}, Algorithm~\ref{alg:pseudo-poly}).
\item Lattice Double Greedy (\textsf{Lattice-DG}): The $1/3$-approximation algorithm for maximizing non-monotone lattice submodular functions~\cite{Gottschalk:2015fq}.
\item Double Greedy with a polynomial number of oracle calls with error parameter $\epsilon>0$ (\textsf{Fast-DG}$_{\epsilon}$, Algorithm~\ref{alg:polynomial-oracle}).
\end{itemize}
We measure the efficiency of an algorithm by the number of oracle calls instead of the total running time.
Indeed, the running time for evaluating the input function is the dominant factor of the total running, because objective functions in typical machine learning tasks contain sums over all data points, which is time consuming.
Therefore, we do not consider the polynomial-time algorithm (Theorem~\ref{the:truely-polynomial}) here.

\subsection{Revenue maximization}

\begin{table}[t!]
  \centering
  \caption{Objective values (Our methods are highlighted in gray.)}\label{tab:ov}
  \scalebox{0.75}{
  \begin{tabular}{|l||rrrrr|}

\hline
& \multicolumn{5}{c|}{Adolescent health}\\
& $B=10^2$ & $10^3$ & $10^4$ & $10^5$ & $10^6$\\
\hline
\textsf{SG} & 280.55 & 2452.16 & 7093.73 & 7331.42 & 7331.50 \\
\textsf{DG} & 280.55 & 2452.16 & 7124.90 & 7332.96 & 7331.50 \\
\textsf{Lattice-DG} & 215.39 & 1699.66 & 6808.97 & 6709.11 & 5734.30 \\
\rowcolor{Gray}
\textsf{Fast-DG}$_{0.5}$ & 280.55 & 2452.16 & 7101.14 & 7331.36 & 7331.48 \\
\rowcolor{Gray}
\textsf{Fast-DG}$_{0.05}$ & 280.55 & 2452.16 & 7100.86 & 7331.36 & 7331.48 \\
\rowcolor{Gray}
\textsf{Fast-DG}$_{0.005}$ & 280.55 & 2452.16 & 7100.83 & 7331.36 & 7331.48 \\
\hline
\hline
& \multicolumn{5}{c|}{Advogato}\\
& $B=10^2$ & $10^3$ & $10^4$ & $10^5$ & $10^6$\\
\hline
\textsf{SG} & 993.15 & 8680.87 & 25516.05 & 27325.78 & 27326.01 \\
\textsf{DG} & 993.15 & 8680.87 & 25330.91 & 27329.39 & 27326.01 \\
\textsf{Lattice-DG} & 753.93 & 6123.39 & 24289.09 & 24878.94 & 21674.35 \\
\rowcolor{Gray}
\textsf{Fast-DG}$_{0.5}$ & 993.15 & 8680.87 & 25520.83 & 27325.75 & 27325.98 \\
\rowcolor{Gray}
\textsf{Fast-DG}$_{0.05}$ & 993.15 & 8680.87 & 25520.52 & 27325.75 & 27325.98 \\
\rowcolor{Gray}
\textsf{Fast-DG}$_{0.005}$ & 993.15 & 8680.87 & 25520.47 & 27325.75 & 27325.98 \\
\hline
\hline
& \multicolumn{5}{c|}{Twitter lists}\\
& $B=10^2$ & $10^3$ & $10^4$ & $10^5$ & $10^6$\\
\hline
\textsf{SG} & 882.43 & 7713.07 & 22452.61 & 25743.26 & 25744.02 \\
\textsf{DG} & 882.43 & 7713.07 & 22455.97 & 25751.42 & 25744.02 \\
\textsf{Lattice-DG} & 675.67 & 5263.87 & 20918.89 & 20847.48 & 15001.19 \\
\rowcolor{Gray}
\textsf{Fast-DG}$_{0.5}$ & 882.43 & 7713.07 & 22664.65 & 25743.06 & 25743.88 \\
\rowcolor{Gray}
\textsf{Fast-DG}$_{0.05}$ & 882.43 & 7713.07 & 22658.58 & 25743.06 & 25743.88 \\
\rowcolor{Gray}
\textsf{Fast-DG}$_{0.005}$ & 882.43 & 7713.07 & 22658.07 & 25743.06 & 25743.88 \\
\hline

  \end{tabular}
  }
\end{table}

\begin{figure*}[!t]
  \centering
  \begin{minipage}{.325\hsize}
    \centering
    \subfloat[Adolescent health]{
      \includegraphics[width=1.1\hsize]{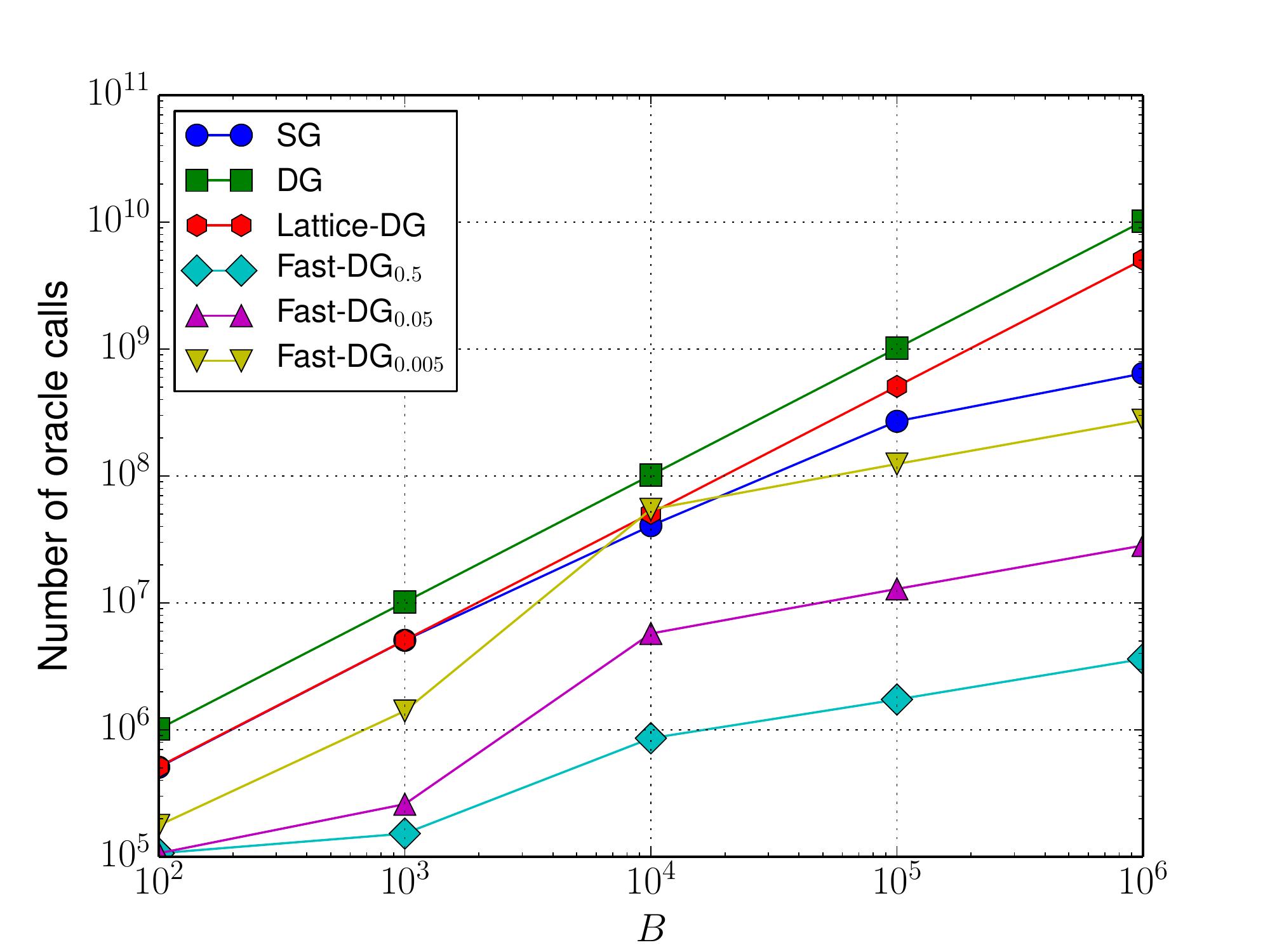}
    }
  \end{minipage}
  \begin{minipage}{.325\hsize}
    \centering
    \subfloat[Advogato]{
    \includegraphics[width=1.1\hsize]{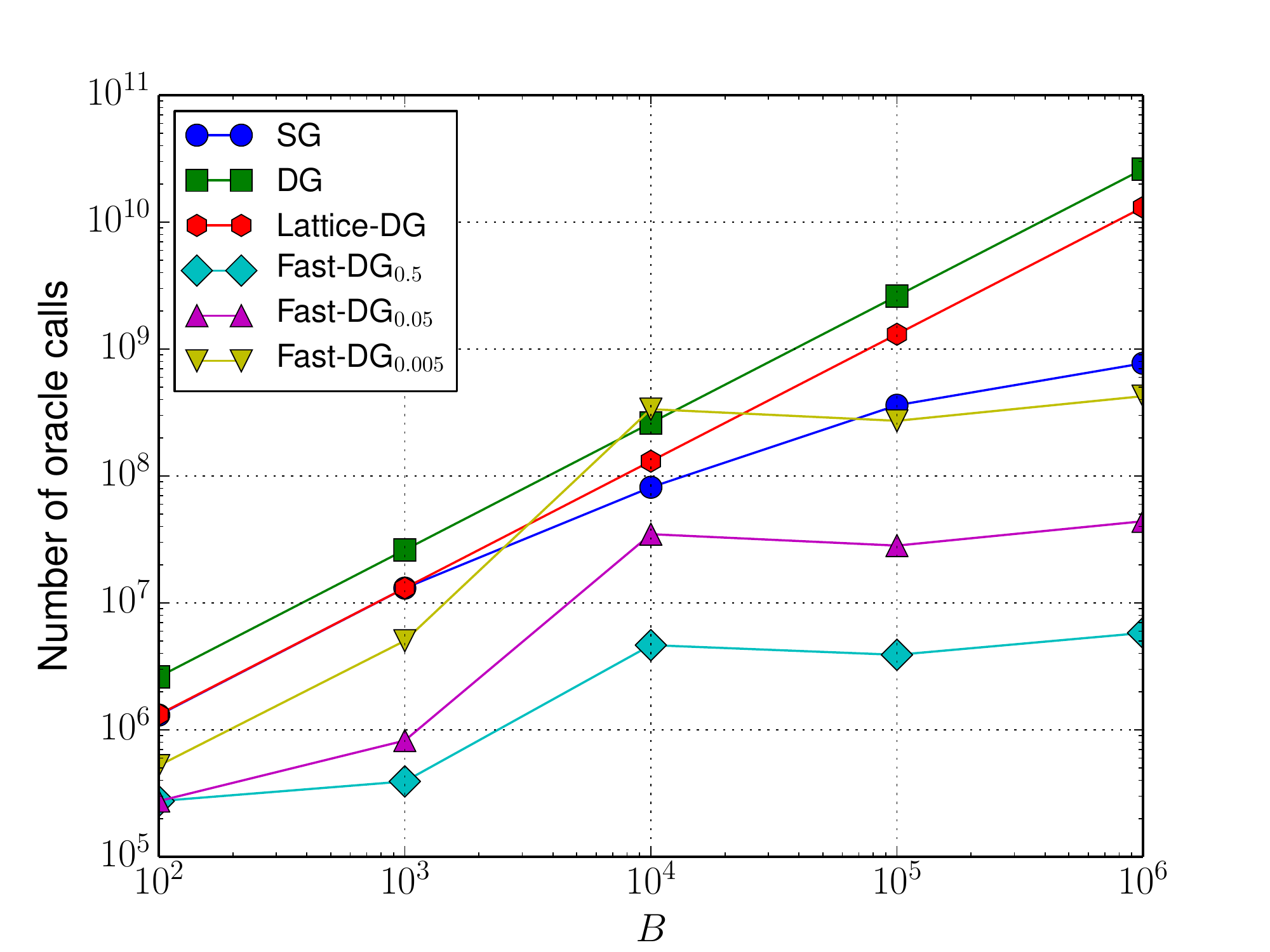}
    }
  \end{minipage}
  \begin{minipage}{.325\hsize}
    \centering
    \subfloat[Twitter lists]{
      \includegraphics[width=1.1\hsize]{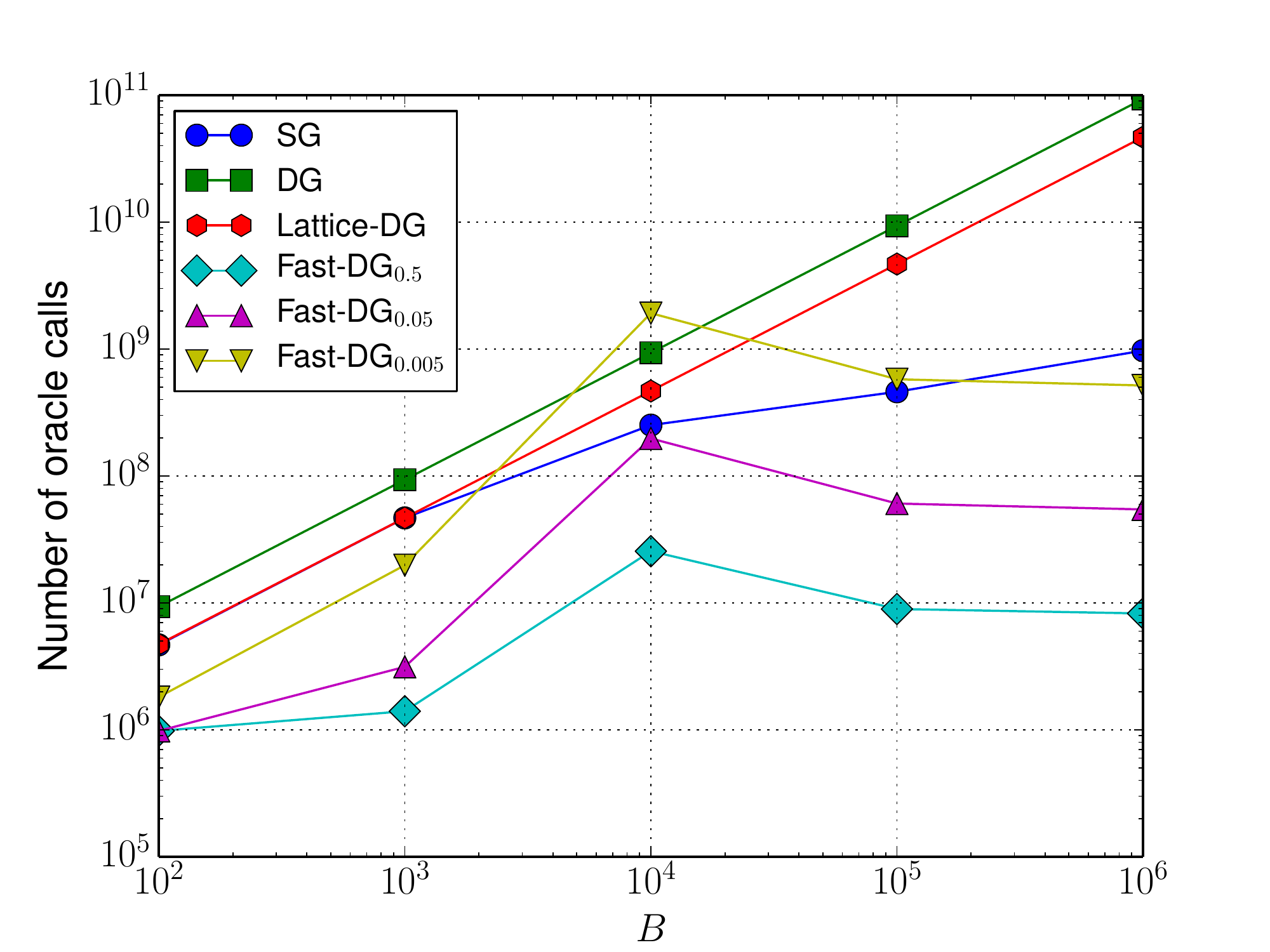}
    }
  \end{minipage}
  \caption{Number of oracle calls}\label{fig:oc-ml}
\end{figure*}

In this application, we consider revenue maximization on an (undirected) social network $G = (V, W)$, where $W = (w_{ij})_{i,j\in V}$ represents the weights of edges.
The goal is to offer for free or advertise a product to users so that the revenue increases through their word-of-mouth effect  on others.
If we invest $x$ units of cost on a user $i \in V$, the user becomes an advocate of the product (independently from other users) with probability $1-(1-p)^x$, where $p \in (0,1)$ is a parameter.
This means that, for investing a unit cost to $i$, we have an extra chance that the user $i$ becomes an advocate with probability $p$.
Let $S \subseteq V$ be a set of users who advocate the product.
Note that $S$ is a random set.
Following a simplified version of the model introduced by~\cite{Hartline:2008fw}, the revenue is defined as $\sum_{i \in S} \sum_{j \in V \setminus S}  w_{ij}$.
Let $f:\bbZ_+^E \to \bbR$ be the expected revenue obtained in this model, that is,
\begin{align*}
  f(\bx) & = \E_{S} \Bigl[\sum_{i \in S} \sum_{j \in V \setminus S}  w_{ij}\Bigr] \\
  & = \sum_{i \in S}  \sum_{j \in V \setminus S} w_{ij}(1-(1-p)^{\bx(i)})(1-p)^{\bx(j)}.
\end{align*}
It is not hard to show that $f$ is non-monotone DR-submodular function (see Appendix~\ref{sec:proof-of-DRsubmod} for the proof).

In our experiment, we used three networks,
Adolescent health (2,539 vertices and 12,969 edges),
Advogato (6,541 vertices and 61,127 edges), and
Twitter lists (23,370 vertices and 33,101 edges), all taken from~\cite{KONECT}.
We regard all the networks as undirected.
We set $p = 0.0001$, and set $w_{ij}=1$ when an edge exists between $i$ and $j$ and $w_{ij}=0$ otherwise.
We imposed the constraint $0 \leq \bx(e) \leq B$ for every $e \in E$, where $B$ is chosen from $\set{10^2,\ldots,10^6}$.

Table~\ref{tab:ov} shows the objective values obtained by each method.
As can be seen, except for \textsf{Lattice-DG}, which is clearly the worst, the choice of a method does not much affect the obtained objective value for all the networks.
Notably, even when $\epsilon$ is as large as $0.5$, the objective values obtained by \textsf{Fast-DG} are almost the same as \textsf{SG} and \textsf{DG}.

Figure~\ref{fig:oc-ml} illustrates the number of oracle calls of each method.
The number of oracle calls of \textsf{DG} and \textsf{Lattice-DG} is linear in $B$, whereas that of \textsf{Fast-DG} slowly grows.
Although the number of oracle calls of \textsf{SG} also slowly grows, it is always orders of magnitude larger than that of \textsf{Fast-DG} with $\epsilon=0.5$ or $\epsilon=0.05$.

In summary, \textsf{Fast-DG}$_{0.5}$ achieves almost the best objective value, whereas the number of oracle calls is two or three orders of magnitude smaller than those of the other methods when $B$ is large.



\section{Conclusions}\label{sec:conclusions}
In this paper, we proposed a polynomial-time $\frac{1}{2+\epsilon}$-approximation algorithm for non-monotone DR-submodular function maximization.
Our experimental results on the revenu maximization problem showed the superiority of our method against other baseline algorithms.

Maximizing a submodular set function under constraints is well studied~\cite{Lee:2009tc,Gupta:2010wj,Chekuri:2014ed,Mirzasoleiman:2016vp}.
An intriguing open question is whether we can obtain polynomial-time algorithms for maximizing DR-submodular functions under constraints such as cardinality constraints, polymatroid constraints, and knapsack constraints.

\subsubsection*{Acknowledgments}
T.~S.~is supported by JSPS Grant-in-Aid for Research Activity Start-up.
Y.~Y.~is supported by JSPS Grant-in-Aid for Young Scientists (B) (No.~26730009), MEXT Grant-in-Aid for Scientific Research on Innovative Areas (No.~24106003), and JST, ERATO, Kawarabayashi Large Graph Project.

\bibliographystyle{aaai}
\bibliography{main}

\appendix
\input{appendix}
\end{document}

%% file: appendix.tex

\section{Proof of Theorem~\ref{the:truely-polynomial}}\label{sec:proof-truely-polynomial}

A key observation to obtain an approximation algorithm with a polynomial time complexity is that the approximate functions $\tilde{g}$ and $\tilde{h}$ used in Algorithm~\ref{alg:polynomial-oracle} are piecewise constant functions.
Hence, while $\bx(e)$ and $\by(e)$ lie on the intervals for which $\tilde{g}$ and $\tilde{h}$, respectively, are constant, the values of $\alpha$ and $\beta$ do not change.
This means that we repeat the same random process in the while loop of Algorithm~\ref{alg:polynomial-oracle} as long as $\bx(e)$ and $\by(e)$ lie on the intervals.
We will show that we can simulate the entire random process in polynomial time.
Because the number of possible values of $\tilde{g}$ and $\tilde{h}$ is bounded by $O(\frac{1}{\epsilon}\log\frac{\Delta}{\delta})$, we obtain a polynomial-time algorithm.

As the model of computation, we assume that we can perform an elementary arithmetic operation on real numbers in constant time, and that we can sample a uniform $[0,1]$ random variable.

The first two ingredients for simulating the random process are sampling procedures for a binomial distribution and a geometric distribution.
For $n \in \bbN$ and $p \in [0,1]$, let $\caB(n,p)$ be the binomial distribution with mean $np$ and variance $np(1-p)$.
For $p \in [0,1]$, let $\caG(p)$ be the geometric distribution with mean $1/p$, that is, $\Pr_{X \sim \caG(p)}[X = k] = (1-p)^{k-1}p$ for $k \geq 1$.
We then have the following:
\begin{lemma}[See, e.g.,~\cite{Devroye:1986zz}]\label{lem:sample-binomial}
  For any $n \in \bbN$, and $p \in [0,1]$, we can sample a value from the binomial distribution $\caB(n,p)$ in $O(\log n)$ time.
\end{lemma}
\begin{lemma}[See, e.g.,~\cite{Devroye:1986zz}]\label{lem:sample-geometric}
  For any $p \in [0,1]$, we can sample a value from the geometric distribution $\caG(p)$ in $O(1)$ time.
\end{lemma}

\begin{algorithm}[t]
  \caption{Subroutine to simulate random processes for Algorithm~\ref{alg:poly}}\label{alg:quick-simulation}
  \begin{algorithmic}[1]
    \REQUIRE{$p \in [0,1]$, $\ell_a,\ell_b,\ell_{a+b} \in \bbZ_+$, $\eta \in (0,1)$.}
    \STATE $q \leftarrow 1-p$.
    \STATE $N \leftarrow O\Bigl(\log \bigl(\frac{1}{\eta}\log (\ell_a+\ell_b+\ell_{a+b})\bigr)\Bigr)$.
    \STATE $a \leftarrow 0$, $b \leftarrow 0$.
    \WHILE{$a < \ell_a$, $b < \ell_b$, and $a+b<\ell_{a+b}$}\label{line:quick-simulation-while}
      \WHILE{$\ell_a -a \leq N$ or $\ell_{a+b} - (a+b) < N$}\label{line:quick-simulation-while-a-or-a+b}
        \STATE $s \leftarrow$ a value sampled from $\caG(q)$.
        \IF{$b + s \leq \ell_b$ and $a+b + s\leq \ell_{a+b}$}
          \STATE $a \leftarrow a + 1$, $b \leftarrow b + s$.
        \ELSE
          \RETURN $(a, \min\set{\ell_b, \ell_{a+b}-a})$.
        \ENDIF
      \ENDWHILE
      \WHILE{$\ell_b -b \leq N$}\label{line:quick-simulation-while-b}
        \STATE $s \leftarrow$ a value sampled from $\caG(p)$.
        \IF{$a + s \leq \ell_a$ and $a+b + s\leq \ell_{a+b}$}
          \STATE $a \leftarrow a + s$, $b \leftarrow b + 1$.
        \ELSE
          \RETURN $(\min\set{\ell_a, \ell_{a+b}-b},b)$.
        \ENDIF
      \ENDWHILE
      \STATE $n \leftarrow \min \set{\floor{\frac{\ell_a-a}{p}}, \floor{\frac{\ell_b-b}{q}}, \ell_{a+b}-(a+b)}$, $m \leftarrow \floor{n/2}$.      \label{line:quick-simulation-else}
      \STATE $s \leftarrow $ a value sampled from $\caB(m,p)$.
      \STATE $a \leftarrow a + s$, $b \leftarrow b + m-s$.
      \IF{$a > \ell_a$, $b>\ell_b$, or $a+b > \ell_{a+b}$}
      \STATE \textbf{Fail}.
      \ENDIF
    \ENDWHILE
    \RETURN{$(a,b)$.}
  \end{algorithmic}
\end{algorithm}

We consider the following random process parameterized by $p \in [0,1]$ and integers $\ell_a,\ell_b,\ell_{a+b} \in \bbZ_+$, which we denote by $\caP(p,\ell_a,\ell_b,\ell_{a+b})$:
We start with $a = b = 0$.
While $a < \ell_a$, $b < \ell_b$, and $a+b< \ell_{a+b}$, we increment $a$ with probability $p$, and $b$ with the complement probability $1-p$.
Note that, in the end of the process, we have $a = \ell_a$, $b = \ell_b$, or $a+b=\ell_{a+b}$.
Let $\caD(p,\ell_a,\ell_b,\ell_{a+b})$ be the distribution of the pair $(a,b)$ generated by $\caP(p,\ell_a,\ell_b,\ell_{a+b})$.

We introduce an efficient procedure (Algorithm~\ref{alg:quick-simulation}) that succeeds in simulating the process $\caP(p,\ell_a,\ell_b,\ell_{a+b})$ with high probability.
To prove the correctness of Algorithm~\ref{alg:quick-simulation}, we use the following form of Chernoff's bound.
\begin{lemma}[Chernoff's bound]\label{lem:chernoff}
  Let $X_1,\ldots,X_n$ be independent random variables taking values in $\set{0,1}$.
  Let $X = \sum_{i=1}^n X_i$ and $\mu = \E[X]$.
  Then, for any $\delta > 1$, we have
  \[
    \Pr[X \geq (1+\delta)\mu] \leq \exp(-\delta \mu/3).
  \]
\end{lemma}

\begin{lemma}\label{lem:quick-simulation}
  We have the following:
  \begin{itemize}
  \item Algorithm~\ref{alg:quick-simulation} succeeds to return a pair $(a,b)$ with probability at least $1-\eta$.
  \item The distribution of the pair $(a,b)$ output by Algorithm~\ref{alg:quick-simulation} is distributed according to $\caD(p,\ell_a,\ell_b,\ell_{a+b})$.
  \item The pair $(a,b)$ output by Algorithm~\ref{alg:quick-simulation} satisfies at least one of the following: $a = \ell_a$, $b = \ell_b$, or $a+b=\ell_{a+b}$.
  \item The time complexity of Algorithm~\ref{alg:quick-simulation} is $O\bigl(\log (\frac{1}{\eta}\log\ell) \cdot \log \ell\bigr)$, where $\ell = \ell_a+\ell_b+\ell_{a+b}$.
  \end{itemize}
\end{lemma}
\begin{proof}
  We first note that once (i) $\ell_a -a \leq N$, (ii) $\ell_b - b \leq N$, or (iii) $\ell_{a+b} - (a+b) \leq N$ holds, then we enter the while loop from Line~\ref{line:quick-simulation-while-a-or-a+b} or from~\ref{line:quick-simulation-while-b} until the end of the algorithm.

  We check the first claim.
  Suppose that none of~(i),~(ii), and~(iii) holds.
  Then, we reach Line~\ref{line:quick-simulation-else}.
  Here, we intend to simulate the process $\caP(p,\ell_a,\ell_b,\ell_{a+b})$ to a point where we will increment $a$ and $b$ $m = \floor{n/2}$ times in total.
  By union bound and Chernoff's bound, the probability that we fail can be bounded by
  \begin{align*}
    & \Pr\Bigl[s \geq \ell_a - a
    \vee
    s \geq \ell_b - b
    \vee
    s \geq \ell_{a+b} - (a +b)\Bigr] \\
    \leq &
    \Pr[s \geq \ell_a - a]
    +
    \Pr[s \geq \ell_b - b]
    +
    \Pr[s \geq \ell_{a+b} - (a +b)] \\
    \leq &
    \exp\Bigl(-\frac{\ell_a - a - pm}{3}\Bigr)
    + \exp\Bigl(-\frac{\ell_b - b - qm}{3}\Bigr)  \\
    & \quad + \exp(-\ell_{a+b} - (a+b) - m)\\
    \leq &
    \exp\Bigl(-\frac{\ell_a - a}{6}\Bigr)
    + \exp\Bigl(-\frac{\ell_b - b}{6}\Bigr) \\
    & \quad + \exp\Bigl(-\frac{\ell_{a+b} - (a+b)}{2}\Bigr)\\
    \leq & 3\exp(-O(N)).
  \end{align*}

  When we do not fail, at least one of the following three values shrinks by half: $\ell_a-a$, $\ell_b-b$, and $\ell_{a+b}-(a+b)$.
  Hence, after $O(\log (\ell_a+\ell_b+\ell_{a+b}))$ iterations of the while loop (from Line~\eqref{line:quick-simulation-while}), at least one of~(i),~(ii), or~(iii) is satisfied.
  Once this happens, we do not fail and output a pair $(a,b)$.
  By union bound, the failure probability is at most
  \[
    3\exp(-O(N))\cdot O(\log(\ell_a+\ell_b+\ell_{a+b}))
    \leq \delta
  \]
  By choosing the hidden constant in $N$ large enough.

  Next, we check the second claim.
  From the argument above, as long as none of (i),~(ii), or~(iii) is satisfied, we exactly simulate the process $\caP(p,\ell_a,\ell_b,\ell_{a+b})$.
  Hence, suppose that (i) is satisfied.
  Then, what we does is until $a$ reaches $\ell_a$, we sample $s$ from the geometric distribution $\caG(q) = \caG(1-p)$, and if $b + s \leq \ell_b$ and $a+b+s \leq \ell_{a+b}$, then we update $a$ by $a+1$ and $b$ by $b+s$, and otherwise we output the pair $(a, b = \min\set{\ell_b,\ell_{a+b}-a})$.
  If $a$ reaches $\ell_a$, then we output the pair $(a = \ell_a,b)$.
  This can be seen as an efficient simulation of the process $\caP(p,\ell_a,\ell_b,\ell_{a+b})$.
  The case (ii) or (iii) is satisfied can be analyzed similarly, and the second claim holds.

  The third and fourth claims are obvious from the definition of the algorithm.
\end{proof}

\begin{algorithm}[t]
  \caption{Polynomial time approximation algorithm}\label{alg:poly}
  \begin{algorithmic}[1]
    \REQUIRE{$f:\bbZ_+^E \to \bbR_+$, $\bB \in \bbZ^E$, $\epsilon > 0$.}
    \ENSURE{$f$ is DR-submodular}
    \STATE $\bx \leftarrow \bzero$, $\by \leftarrow \bB$.
    \FOR{$e \in E$}
      \STATE Define $g,h:\set{0,1,\ldots,B-1} \to \bbR$ as $g(b) = f(\bchi_e \mid \bx+b\bchi_e)$ and $h(b) = f(- \bchi_e \mid \by-b\bchi_e)$. \hspace{-1em}
      \STATE Let $\tilde{g}$ and $\tilde{h}$ be approximations to $g$ and $h$, respectively, given by Algorithm~\ref{alg:approximate}.
      \WHILE{$\bx(e) < \by(e)$}
        \STATE $\alpha \leftarrow \tilde{g}(\bx(e))$ and $\beta \leftarrow \tilde{h}(\bB(e) - \by(e))$.
        \IF{$\beta = 0$}
          \STATE $\bx(e) \leftarrow \by(e)$ and \textbf{break}.
        \ELSIF{$\alpha < 0$}
          \STATE $\by(e) \leftarrow \bx(e)$ and \textbf{break}.
        \ELSE
        \STATE $s \leftarrow \max \set{b \mid \tilde{g}(b)=\alpha}$, $t \leftarrow \bB(e) - \max \set{b \mid \tilde{h}(b)=\beta}$.
        \STATE Call Algorithm~\ref{alg:quick-simulation} with $p = \frac{\alpha}{\alpha+\beta}$, $\ell_a = s - \bx(e)$, $\ell_b = \by(e) - t$, $\ell_{a+b} = \by(e)-\bx(e)$, and $\eta = O(\frac{\epsilon^2}{(2+\epsilon)|E|\log(\Delta/\delta)})$.\label{line:poly-call-quick-simulation}
        \IF{Algorithm~\ref{alg:quick-simulation} returned a pair $(a,b)$}
          \STATE $\bx(e) \leftarrow \bx(e) + a$, $\by(e) \leftarrow \by(e) + b$.
        \ELSE
          \RETURN $\bzero$.
        \ENDIF
        \ENDIF
      \ENDWHILE
    \ENDFOR
    \RETURN $\bx$.
  \end{algorithmic}
\end{algorithm}

Our idea for simulating Algorithm~\ref{alg:polynomial-oracle} efficiently is as follows.
Suppose we have $\alpha = \tilde{g}(\bx(e))$ and $\beta = \tilde{h}(\bB(e)-\by(e))$ for the current $\bx$ and $\by$.
Let $s = \max \set{b \mid \tilde{g}(b)=\alpha}$ and $t = \bB(e) - \max \set{b \mid \tilde{h}(b)=\beta}$.
Then, $\tilde{g}$ and $\tilde{h}$ are constant in the intervals $[\bx(e),\ldots,s]$ and $[\bB(e) - t,\ldots,\bB(e)-\by(e)]$, respectively.
By running Algorithm~\ref{alg:quick-simulation} with $p = \alpha/(\alpha+\beta)$, $\ell_a = s - \bx(e)$, $\ell_b = t - \by(e)$, and $\ell_{a+b} = \by(e) - \bx(e)$, we can simulate Algorithm~\ref{alg:polynomial-oracle} to a point where at least one of the following happens: $\bx(e)$ reaches $s$, $\by(e)$ reaches $t$, or $\bx(e)$ is equal to $\by(e)$.
When Algorithm~\ref{alg:quick-simulation} failed to output a pair, we output an arbitrary feasible solution, say, the zero vector $\bzero$.
Algorithm~\ref{alg:poly} presents a formal description of the algorithm.





\begin{proof}[Proof of Theorem~\ref{the:truely-polynomial}]
  We first analyze the failure probability.
  Since the number of possible values of $\tilde{g}$ and $\tilde{h}$ is bounded by $O(\frac{1}{\epsilon}\log\frac{\Delta}{\delta})$ for each $e \in E$, we call Algorithm~\ref{alg:quick-simulation} $O(\frac{|E|}{\epsilon}\log\frac{\Delta}{\delta})$ times by the third claim of Lemma~\ref{lem:quick-simulation}.
  Hence, by the first claim of Lemma~\ref{lem:quick-simulation} and union bound, the failure probability is at most $\frac{\epsilon}{2+2\epsilon}$ if the hidden constant in $\eta$ at Line~\ref{line:poly-call-quick-simulation} is chosen to be small enough.

  Let $\caD$ and $\caD'$ be the distributions of outputs from Algorithms~\ref{alg:polynomial-oracle} and~\ref{alg:poly}, respectively.
  Conditioned on the event that Algorithm~\ref{alg:quick-simulation} does not fail (and hence we output $f(\bzero)$), $\caD$ exactly matches $\caD'$ by the second claim of Lemma~\ref{lem:quick-simulation}.

  By letting $\bo$ be the optimal solution, we have, by Theorem~\ref{the:polynomial-oracle}, that
  \[
    \E_{\bx \sim \caD'}f(\bx) \geq \frac{2+\epsilon}{2+2\epsilon} \cdot \E_{\bx \sim \caD}f(\bx) \geq \frac{1}{2+2\epsilon}f(\bo).
  \]
  Hence, we have the approximation factor of $\frac{1}{2+\epsilon}$.

  The number of oracle calls is exactly the same as Algorithm~\ref{alg:polynomial-oracle}.
  The total time spent inside Algorithm~\ref{alg:quick-simulation} is
  \begin{align*}
    & O\Bigl(\frac{|E|}{\epsilon}\log\frac{\Delta}{\delta}\Bigr) \cdot O\Bigl(\log\Bigl(\frac{|E|}{\epsilon}\log\frac{\Delta}{\delta} \log\|\bB\|_\infty\Bigr) \log\|\bB\|_\infty\Bigr) \\
    & =
    \widetilde{O}\Bigl(\frac{|E|}{\epsilon}\log\frac{\Delta}{\delta}\log\|\bB\|_\infty\Bigr)
  \end{align*}
  by the fourth claim of Lemma~\ref{lem:quick-simulation}.
  Because evaluating $\tilde{g}$ and $\tilde{h}$ takes $O(\log \|\bB\|_\infty)$ time and computing $s$ and $t$ also takes $O(\log \|\bB\|_\infty)$ time, we need
  \[
    O\Bigl(\frac{|E|}{\epsilon}\log\frac{\Delta}{\delta}\Bigr) \cdot O(\log^2 \|\bB\|_\infty)
    =
    O\Bigl(\frac{|E|}{\epsilon}\log\frac{\Delta}{\delta}\log^2 \|\bB\|_\infty\Bigr)
  \]
  time.
  Summing them up, we obtain the stated time complexity.
\end{proof}

We again note that, even if the given DR-submodular function $f:\bbZ_+^E \to \bbR$ is not non-negative, we can obtain a $\frac{1}{2+2\epsilon}$-approximation algorithm for~\eqref{eq:max-DRsfm-g}, as stated in Remark~\ref{rem:non-negative}.

\section{DR-submodularity of functions used in experiments}\label{sec:proof-of-DRsubmod}
In this section, we will see that the objective function used in Section~\ref{sec:experiments} is indeed DR-submodular.
Recall that our objective function of revenue maximization is as follows:
\begin{align*}
  f(\bx)  = \sum_{i \in S}  \sum_{j \in V \setminus S} w_{ij}(1-(1-p)^{\bx(i)})(1-p)^{\bx(j)},
\end{align*}
where $w_{ij}$ is a nonnegative weight and $p \in [0,1]$ is a parameter.
Since DR-submodular functions are closed under nonnegative linear combination, it suffices to check that
\begin{align*}
    g(\bx)  = (1-q^{\bx(i)})q^{\bx(j)}
\end{align*}
is DR-submodular, where $q = 1-p$.
To see the DR-submodularity of $g$, we need to check that
\begin{align*}
    g(\bchi_i \mid \bx + \bchi_j) &\leq g(\bchi_i \mid \bx) \quad (\bx \in \bbR^E_+, i, j \in E).
\end{align*}
Note that $i$ and $j$ may be identical.
By direct algebra,
\begin{align*}
    g(\bchi_i \mid \bx)
    &= (1 - q^{\bx(i)+1})q^{\bx(j)} - (1 - q^{\bx(i)})q^{\bx(j)} \\
    &= q^{\bx(i)+\bx(j)}(1-q), \\
    g(\bchi_i \mid \bx + \bchi_j)
    &= q^{\bx(i)+\bx(j)+1}(1-q) \\
    &= q g(\bchi_i \mid \bx).
\end{align*}
Since $q \in [0,1]$, we obtain $g(\bchi_i \mid \bx + \bchi_j) \leq g(\bchi_i \mid \bx)$.